%% file: Wiopt2025.tex
\documentclass[10pt,conference]{IEEEtran}
\usepackage{fixltx2e}
\usepackage{cite}
\usepackage{url}
\usepackage{mathrsfs}
\usepackage{color}
\usepackage{float}
\usepackage{balance}
\usepackage[caption=false]{subfig}

\ifCLASSINFOpdf
   \usepackage[pdftex]{graphicx}
   \graphicspath{{Figs/}}
   \DeclareGraphicsExtensions{.pdf,.jpeg,.png}
\else
\fi
\usepackage[cmex10]{amsmath}
\usepackage{amsmath}
\usepackage{amssymb}
\usepackage{amsthm}
\usepackage{amsfonts}
\usepackage{bm}
\usepackage{xfrac}
\usepackage{empheq}
\usepackage[normalem]{ulem} % either use this (simple) or
\usepackage{soul} % use this (many fancier options)
\usepackage{mathtools}
\usepackage{nicefrac}
\usepackage{cleveref}

\newtheorem{theorem}{Theorem}

\newtheorem{example}{Example}

\newtheorem{lemma}{Lemma}

\newtheorem{corollary}{Corollary}
\newtheorem{remark}{Remark}
\theoremstyle{definition}

\renewcommand{\qed}{\hfill $\blacksquare$}
\newenvironment{sproof}{\noindent{ \emph{ Sketch of proof:}}}{\qed\bigskip}
\usepackage[a4paper,bindingoffset=0.2in,%
left=1in,right=1in,top=1in,bottom=1.58in,%
footskip=.25in]{geometry}
\graphicspath{{figs/}}
\begin{document}
	\newgeometry{left=0.625in,right=0.625in,top=0.7in,bottom=0.99in}
	
    \title{Information-Theoretic Fairness with A Bounded Statistical Parity Constraint} 
\vspace{-5mm}
\author{
		\IEEEauthorblockN{ \vspace*{0.5em}
			\IEEEauthorblockA{Amirreza Zamani, Abolfazl Changizi, Ragnar Thobaben, Mikael Skoglund\\
                              Division of Information Science and Engineering, KTH Royal Institute of Technology \\
                              %$^\dagger$Nokia Bell Labs, Holmdel, NJ, USA\\
				Email: amizam@kth.se,
                changizi@kth.se, ragnart@kth.se, skoglund@kth.se \protect }}%\vspace*{-2.1em}
		}
	\maketitle

\begin{abstract}
    In this paper, we study an information-theoretic problem of designing a fair representation that attains bounded statistical (demographic) parity. More specifically, an agent uses some useful data $X$ to solve a task $T$. Since both $X$ and $T$ are correlated with some sensitive attribute or secret $S$, the agent designs a representation $Y$ that satisfies a bounded statistical parity and/or privacy leakage constraint, that is, such that $I(Y;S) \leq \epsilon$. Here, we relax the perfect demographic (statistical) parity and consider a bounded-parity constraint. 
    In this work, we design the representation $Y$ that maximizes the mutual information $I(Y;T)$ about the task while satisfying a bounded compression (or encoding rate) constraint, that is, ensuring that $I(Y;X) \leq r$. Simultaneously, $Y$ satisfies the bounded statistical parity constraint $I(Y;S) \leq \epsilon$.  %Second, inspired by the Conditional Fairness Bottleneck problem, we consider a design desiderata where we want to maximize the information $I(Y;T|S)$ that the representation contains about the task which is not shared by the sensitive attribute or secret, while constraining the amount of irrelevant information, that is, ensuring that $I(Y;X|T,S) \leq r$. 

    To design $Y$, we use extended versions of the Functional Representation Lemma and the Strong Functional Representation Lemma which are based on randomization techniques and study the tightness of the obtained bounds in special cases. The main idea to derive the lower bounds is to use randomization over useful data $X$ or sensitive data $S$. 
Considering perfect demographic parity, i.e., $\epsilon=0$, we improve the existing results (lower bounds) by using a tighter version of the Strong Functional Representation Lemma and propose new upper bounds. We then propose upper and lower bounds for the main problem and show that allowing non-zero leakage can improve the attained utility. Finally, we study the bounds and compare them in a numerical example.

    The problem studied in this paper can also be interpreted as one of code design with bounded leakage and bounded rate privacy considering the sensitive attribute as a secret.

\end{abstract}
\section{Introduction}
\input{contents/intro}

\section{Preliminaries and Background}
\label{sec:background}
\input{contents/Background}

\section{Problem Formulation} 
\label{sec:system}
\input{contents/system}
\section{Main Results}
\label{sec:resul}
\input{contents/main_results}

\section{Conclusion}
\input{contents/conclusion}
%   \vspace{-6mm}
\section{Acknowledgment}
This work is supported in part by the Swedish Innovation Agency through the SweWIN center.
%\vspace{-1mm}
\section*{Appendix A}
\input{contents/Appendix}
\bibliographystyle{IEEEtran}
{\balance \bibliography{IEEEabrv,Wiopt2025}}
\end{document}

%% file: contents/intro.tex
In this paper, we consider the scenario illustrated in~\Cref{fig:ISITsys}, where we want to use some observable useful data $X$ to draw inferences or make some decision about a certain task $T$. Here, we assume that both the useful data and the task are arbitrarily correlated with some sensitive attribute or secret $S$. For example, the data  $X$ could represent a person's resume, the task $T$ could determine whether this person should be employed in a position or not, and the sensitive attribute $S$ could correspond to the person's gender or disabilities.
 
\begin{figure}[t]
    \centering
    \includegraphics[width=0.4\textwidth]{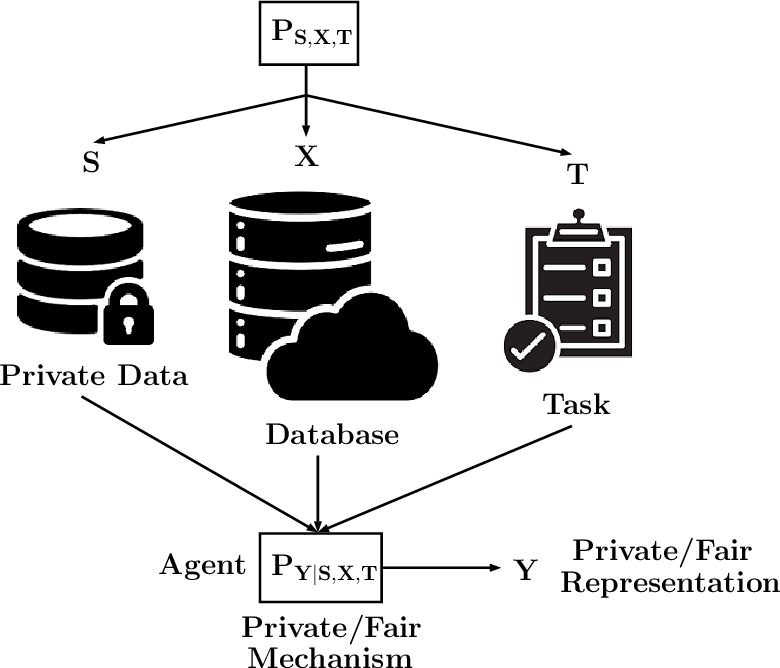}
    \caption{\looseness=-1 Data representation with a bounded statistical (demographic) parity or privacy leakage constraint. The goal is to design a representation $Y$ of the data $X$ that is useful for the task $T$, is compressed, and leaks within a controlled threshold of the sensitive attribute or secret $S$.}
    %\vspace{-3mm}
    \label{fig:ISITsys}
\end{figure}
\looseness=-1 As outlined in \cite{AmirITW2024,vari}, it is important to ensure that decisions are not unfair and/or that inferences do not violate privacy leakage constraints. To this end, we can design a \emph{private} or \emph{fair} representation $Y$ of the data, that is, representations that contain the maximum possible information about the task $T$ while leaking no more than an acceptable threshold about the sensitive attribute or secret 
 $S$~\cite{AmirITW2024,gun, vari,zhao2022, zhao2019,zemel, gronowski2023classification,hardt, king3, borz, khodam, Khodam22,Yanina1,makhdoumi, yamamoto1988rate, sankar, Calmon2,fairAsoodeh}. As discussed in~\cite{vari,AmirITW2024,fairAsoodeh}, privacy and fairness problems are very similar, and their intersections can be utilized in broad areas such as classification algorithms. More specifically, in \cite{AmirITW2024}, the design of representations $Y$ that contain no information about the sensitive attribute or secret $S$, i.e., $I(Y;S) = 0$, has been considered. %This has different meanings in different communities.
In the \emph{information-theoretic privacy} literature, the independence of $Y$ and $S$ is known as \emph{perfect privacy} or \emph{perfect secrecy}~\cite{borz, Yanina1}.
In the \emph{algorithmic fairness} literature, the independence of the algorithm's output $Y$ and the sensitive attribute $S$ is known as \emph{perfect demographic (or statistical) parity} \cite{vari,zhao2022,zhao2019,zemel}. %Here, we will refer to this condition as \emph{perfect demographic parity}. 

As outlined in \cite{shannon,khodam,Yanina1,fairAsoodeh,king3}, there are many cases where perfect privacy (secrecy) or perfect demographic (statistical) parity is not achievable and we need to relax the restriction. Specifically, there are applications where perfect privacy (secrecy) leads to zero utility while relaxing it and allowing a small leakage can result in a significant gain \cite{khodam}; see, for example, Example 1. Therefore, in this paper, we relax the perfect demographic parity constraint in \cite{AmirITW2024}, and consider a bounded statistical parity condition.
This relaxation unlocks new design strategies, which lead to improved utilities and are not applicable under the perfect demographic parity or perfect privacy constraint (e.g., see Remark 6).

\looseness=-1 Under the outlined requirements on the fairness and/or privacy of the representations, we consider a problem that solves a trade-off between utility, compression or encoding rate and statistical parity:
    To ensure that the representation $Y$ has as much impact as possible, we maximize the information it contains about the task $T$. Moreover, we impose a minimum level of compression $r$ to the data representation, that is, $I(Y;X) \leq r$. Finally, the representation $Y$ should satisfy the bounded statistical parity constraint $I(S;Y)\leq \epsilon$. This optimization problem is described in~\eqref{eq:problem_1}.
    %

%\begin{figure}[t]
%	\centering
%	\includegraphics[width=0.4\textwidth]{Figs/cfb-vd-perfect.pdf}
%	\caption{Information diagram~\cite{yeung1991new} of the CFB with perfect demographic parity. In light gray, we show the relevant information about the task $T$, not shared by the sensitive attribute $S$, that we want the representation $Y$ to maximize. In dark gray, we show the irrelevant information about the data $X$ that we want to constrain. Contrary to the standard CFB~\cite{vari}, we enforce that the representation contains no information about the sensitive attribute (perfect demographic parity or perfect privacy).}
%	\label{fig:cfb-perfect}
%\end{figure}

\subsection{Prior work on fairness mechanism design}

The concept of \emph{fair representations} was introduced by Zemel et al.~\cite{zemel}, marking a significant advancement in the field of algorithmic fairness through the expressive capabilities of deep learning. Subsequent research has been predominantly characterized by adversarial~\cite{zemel, edwards2016censoring, zhao2019} and variational~\cite{vari, creager2019flexibly, louizos2015variational, gupta2021controllable} methodologies. The theoretical exploration of the trade-offs between utility and fairness is further elaborated in~\cite{zhao2022}.

In \cite{vari}, the authors study an information-theoretic fairness problem. They introduce the CFB (Conditional Fairness Bottleneck), which quantifies the trades-off between the utility, fairness, and \emph{compression} of the representations in terms of mutual information. This additional criterion has since been used in subsequent studies~\cite{gupta2021controllable, de2022funck, gronowski2023classification}.~\cite{gupta2021controllable, de2022funck, gronowski2023classification}. More specifically, in \cite{vari}, the representation design is based on a variational approach. In contrast with \cite{vari}, in \cite{AmirITW2024}, the authors propose a simple and constructive theoretical approach to design fair representations with perfect demographic parity and/or private representations with perfect privacy. They provide upper and lower bounds on a trade-off between the utility, fairness, and compression rate of the representation. To attain the lower bounds, they use randomization techniques and extended versions of \emph{Functional Representation Lemma} (FRL) and \emph{Strong Functional Representation Lemma} (SFRL) derived in \cite{king3}. Furthermore, they have shown that the lower bounds can be tight under specific assumptions. In \cite{fairAsoodeh}, a binary classification problem subject to both differential privacy and fairness constraints is studied.

\subsection{Prior work on privacy mechanism design}

In \cite{borz}, the authors study the trade-off between privacy and utility in terms of mutual information. They show that the optimal privacy mechanism under perfect privacy can be obtained as the solution to a linear problem. In \cite{khodam, Khodam22}, this is generalized by relaxing the perfect privacy assumption and allowing some small bounded privacy leakage. In both \cite{khodam} and \cite{Khodam22}, point-wise (per-letter) measures have been used in the leakage constraint. Furthermore, the privacy design is based on the information geometry concept which allows them to study the problem geometrically.
The notion of the \emph{Privacy Funnel} is introduced in~\cite{makhdoumi}, where the trade-off between privacy and utility considers the log-loss as the privacy and distortion metrics. %ultimately leading again to the mutual information. %An extension of this problem, the \emph{Conditional Privacy Funnel}, was introduced in~\cite{vari}. 

In \cite{Calmon2}, the authors studied the fundamental limits of the privacy and utility trade-off under an estimation-theoretic approach. In \cite{yamamoto, sankar}, the trade-off between privacy and utility is studied considering the equivocation as the privacy measure and the expected distortion as a utility. In \cite{Yanina1}, the authors introduce the notion of \emph{secrecy by design} and apply it to privacy mechanism design and lossless compression problems. They find bounds on the trade-off between perfect privacy and utility using the \emph{Functional Representation Lemma} (FRL). In \cite{king3}, these results are generalized to an arbitrary privacy leakage. To attain privacy mechanism, extensions of the FRL and SFRL have been used. In \cite{zamani2025variable}, a compression design problem under the secrecy constraints is studied and bounds on the average length of the encoded message are derived. In \cite{9457633}, fundamental limits of private data disclosure are studied, where the goal is to minimize leakage under utility constraints with non-specific tasks. This result is generalized in \cite{zamani2023multi}. The concept of lift is studied
in \cite{zarab2} which represents the likelihood ratio between the posterior and prior beliefs
concerning sensitive features within a data set. 
A privacy mechanism design in a cache-aided network has been studied in \cite{amircache}.
%In \cite{Yanina1}, the authors also studied the average length of the encoded message. %and in~\cite{Yanina2}, they studied the relationships between the shared key, secrecy, and compression considering perfect secrecy, secrecy by design, maximal leakage, mutual information, and differential privacy for lossless compression.
 
\subsection{Contributions}

In this paper, we propose a simple and constructive theoretical approach to design fair representations with bounded statistical parity and/or private representations with bounded privacy leakage.
We study the problem of finding a trade-off between utility, fairness and/or privacy and compression, which is outlined earlier.
%We consider two problems depending on the definition of the utility and compression constraint, which we outlined above.

\looseness=-1 The design of the mechanism yielding the representations is based on the extensions of the FRL and the SFRL derived in~\cite{king3}. As outlined in~\cite{king3}, the main idea for extending FRL~\cite{Yanina1} and SFRL~\cite{cheuk} is to use a randomization technique, which is also employed in~\cite{warner1965randomized}. 
Furthermore, it is shown that both lemmas are constructive and simple, which helps in obtaining mechanism designs that can be optimal in special cases. We follow a similar approach as used in \cite{king3} and \cite{Yanina1} to obtain simple mechanisms for fair and private representations. 
We emphasize that to extend the SFRL, in this paper, we use a tighter version of it which is introduced in \cite{new_functional}.
Finally, we compare the obtained designs in different scenarios and show that they are optimal in special cases. \\
In more details, we can summarize our contribution as follows
\begin{itemize}
\item In Section \ref{sec:background}, we provide a background on the important lemmas such as the FRL, SFRL, and their extended versions. Furthermore, we provide an overview of the previous fairness problem studied in \cite{AmirITW2024}. 
    \item In Section \ref{sec:system}, we introduce the fair/private representation design problem with a bounded statistical (demographic) parity. As outlined earlier, we relax the perfect demographic constraint in \cite{AmirITW2024} and consider a bounded condition.
    \item In Section \ref{sec:resul}, we first improve the previous results in \cite{AmirITW2024} by using a tighter version of the SFRL \cite{new_functional} and also provide new upper bounds. Next, we study the problem with a bounded constraint. We provide upper and lower bounds on the trade-off and analyze their tightness in special cases. 
    %To obtain the lower bounds we use a randomization technique used in the extended versions of the FRL and SFRL. We propose three lower bounds that extend the previous results in \cite{AmirITW2024} and introduce two new lower bounds. 
    We show that relaxing either the zero-leakage or the perfect statistical parity constraint leads to new design strategies that were previously inapplicable.
    Finally, in a numerical example, we compare the obtained bounds.\\
    The paper is concluded in Section V.
\end{itemize}

%% file: contents/Background.tex
In this section, we first recall the FRL \cite[Lemma~1]{Yanina1}, a tighter version of the SFRL \cite[Theorem 16]{new_functional} and the extended version FRL \cite{king3}. We then extend the new version of SFRL using \cite{king3}.
\begin{lemma}\label{lemma1} (FRL \cite[Lemma~1]{Yanina1}):
 	For any pair of RVs $(X,Y)$, there exists a RV $U$ supported on $\mathcal{U}$ such that $X$ and $U$ are independent, i.e., we have
 	$
 	I(U;X)=0,
 $
 	$Y$ is a deterministic function of $(U,X)$, i.e., we have
 	$
 	H(Y|U,X)=0,\label{c2}
 	$
 	and 
 	$
 	|\mathcal{U}|\leq |\mathcal{X}|(|\mathcal{Y}|-1)+1.
 	$
 \end{lemma}
  \begin{lemma}\label{lemma2} (A Tighter SFRL \cite[Theorem~16]{new_functional}):
 	For any pair of RVs $(X,Y)$, there exists a RV $U$ supported on $\mathcal{U}$ such that $X$ and $U$ are independent, i.e., we have
 	$
 	I(U;X)=0,
 	$
 	$Y$ is a deterministic function of $(U,X)$, i.e., we have 
 	$
 	H(Y|U,X)=0,
 	$ and
 	$I(X;U|Y)$ can be upper bounded as follows
 	$
 	I(X;U|Y)\leq \log(I(X;Y)+5.51)+1.06.
 	$
 \end{lemma}
 \begin{lemma}\label{lemma3} (Extended Functional Representation Lemma (EFRL)):
	For a pair of RVs $(X,Y)$ and any $0\leq\epsilon< I(X;Y)$, there exists a RV $U$ defined on $\mathcal{U}$ such that the leakage between $X$ and $U$ is equal to $\epsilon$, i.e., we have
	$
	I(U;X)= \epsilon,
	$
	$Y$ is a deterministic function of $(U,X)$, i.e., we have  
	$
	H(Y|U,X)=0,
	$
	and 
	$
	|\mathcal{U}|\leq \left[|\mathcal{X}|(|\mathcal{Y}|-1)+1\right]\left[|\mathcal{X}|+1\right].
	$
\end{lemma}
Next, we improve the ESFRL in \cite{king3} using Lemma \ref{lemma2} instead of \cite[Theorem~1]{cheuk}. 
\begin{lemma}\label{lemma4}
	For a pair of RVs $(X,Y)$ and any $0\leq\epsilon< I(X;Y)$, there exists a RV $U$ defined on $\mathcal{U}$ such that the leakage between $X$ and $U$ is equal to $\epsilon$, i.e., we have
	$
	I(U;X)= \epsilon,
	$
	$Y$ is a deterministic function of $(U,X)$, i.e., we have 
	$
	H(Y|U,X)=0,
	$ and 
	$I(X;U|Y)$ can be  upper bounded as follows 
	$
	I(X;U|Y)\leq \alpha H(X|Y)+(1-\alpha)\left[ \log(I(X;Y)+5.51)+1.06\right],
	$
	where $\alpha =\frac{\epsilon}{H(X)}$.
\end{lemma}
\begin{proof}
    The proof follows similar arguments as \cite[Lemma 4]{king3}. The only difference is to use Lemma \ref{lemma2} instead of the SFRL in \cite{cheuk}. 
\end{proof}
The main idea of constructing 
$U$, which satisfies the conditions in Lemmas \ref{lemma3} and \ref{lemma4}, is to add a randomized response to the output of Lemma 1 and 2, respectively. The output refers to a random variable that is produced by the FRL or the tighter version of SFRL and satisfies the corresponding constraints.
This idea has been used to find fair and private representations of a dataset in \cite{AmirITW2024}. In more detail, in \cite{AmirITW2024}, the following problems are studied:
\begin{align}
g^{r}(P_{S,X,T})&=\sup_{\begin{array}{c} 
	\substack{P_{Y|S,X,T}:I(Y;S)=0,\\ \ \ \ \ \ \ \ \  I(X;Y)\leq r }
	\end{array}}I(Y;T),
\label{eq:problem_11}\\
g^{r}_c(P_{S,X,T})&=\sup_{\begin{array}{c} 
	\substack{P_{Y|S,X,T}:I(Y;S)=0,\\ \ \ \ \ \ \ \ \ \ I(X;Y|S,T)\leq r }
	\end{array}}I(Y;T|S),
\label{eq:problem_22}
\end{align}
where $X$, $S$, $T$, and $Y$ denote the dataset, sensitive attribute, target, and the designed representation, respectively. To find fair representations, we used randomization techniques similar to those in Lemma \ref{lemma3} and Lemma \ref{lemma4}, which lead to lower bounds on \eqref{eq:problem_11} and \eqref{eq:problem_22}. Furthermore, this randomization technique has also been applied in the design of differentially private mechanisms \cite{erlingsson2014rappor,asoodeh2021three,alghamdi2023saddle,dwork2010boosting}. A similar randomization technique has been used in \cite{zamani2023private} to address private compression design problems.

%% file: contents/system.tex
In this section, we introduce the problem of designing a compressed representation of data with a bounded statistical (demographic) parity and/or privacy leakage constraint. Let the sensitive data $S$, shared data (useful data) $X$, and task $T$ be discrete random variables (RVs) defined on alphabets $\mathcal{S}$ of finite cardinality $\left| \mathcal{S} \right|
$, $\mathcal{X}$ of finite or countably infinite cardinality $\left| \mathcal{X} \right|$, and $\mathcal{T}$ of finite or countably infinite cardinality $\left| \mathcal{T} \right|
$, respectively. The marginal probability distributions of $S$, $X$, and $T$ are denoted by $P_S \in \mathbb{R}^{|\mathcal{S}|}$, $P_X \in \mathbb{R}^{|\mathcal{X}|}$, and $P_T \in \mathbb{R}^{|\mathcal{T}|}$, respectively. Furthermore, $P_{S,X,T} \in \mathbb{R}^{|\mathcal{S}|\times |\mathcal{X}| \times |\mathcal{T}|}$ denotes the joint distribution of $S$, $X$ and $T$.

Let the discrete RV $Y \in \mathcal{Y}$ denote the fair representation. The goal is to design a mapping $P_{Y|S,X,T} \in \mathbb{R}^{|\mathcal{Y}| \times |\mathcal{S}|\times |\mathcal{X}| \times |\mathcal{T}|}$ that maximizes the information it keeps about the task $T$, while maintaining a minimum level of compression $r$ and allowing a certain level of leakage $\epsilon$. As we outlined earlier, in this work, we relax the fairness constraint in \cite{AmirITW2024}, and we allow a bounded leakage, i.e., $I(Y;S)\leq\epsilon$. Hence, the fair/private representation design problem can be stated as follows:

\begin{align}
    g^{r}_{\epsilon}(P_{S,X,T})
    &=\sup_{\begin{array}{c} 
	\substack{P_{Y|S,X,T}:I(Y;S) \leq \epsilon,\\ I(X;Y)\leq r }
	\end{array}}I(Y;T).
    \label{eq:problem_1}
\end{align} 
The condition $I(U;S)\leq \epsilon$ corresponds to the bounded statistical parity or privacy leakage constraint, and $I(X;Y)\leq r$ represents the compression rate constraint.  
\begin{remark}
\normalfont
\looseness=-1 Similar to \cite{AmirITW2024}, we consider the case where $r < H(X)$. Otherwise, this leads to the privacy-utility trade-off problem studied in~\cite{Yanina1, king3}. Using a similar argument, we assume that $\epsilon<H(S)$, otherwise, considering $X$ as a private attribute, we get the problem in~\cite{Yanina1, king3}.
\end{remark}%
\begin{remark}
\normalfont
Similar to \cite{AmirITW2024} and in contrast with \cite{gun}, we assume that $X$, $S$ and $T$ are arbitrarily correlated and we ignored the Markov chain $(S,T)-X-Y$. As we mentioned in \cite{AmirITW2024}, an example where the Markov chain naturally holds is to let both $S$ and $T$ be deterministic functions of $X$. 
\end{remark}
\begin{remark}
\normalfont
\looseness=-1 For $\epsilon = 0$, \eqref{eq:problem_1}
leads to designing a fair representation of data with perfect demographic parity studied in \cite{AmirITW2024}, where upper and lower bounds have been obtained. We improve the results in \cite{AmirITW2024} using a tighter SFRL in \cite{new_functional}, i.e., Lemma \ref{lemma2}, as presented in Theorem \ref{th_prev}.
\end{remark}%
\begin{remark}
\normalfont
Although in this paper we assume that $X$, $S$ and $T$ are discrete RVs, all results can be extended to continuous RVs since Lemma \ref{lemma2} also holds for continuous RVs. For more detail, see \cite[Th. 16.]{new_functional}.
\end{remark}
As we outlined earlier, there are many cases where zero demographic parity/privacy leakage leads to zero utility.
Next, we provide an example where this is the case for perfect demographic parity, i.e., $I(U;S)=0$.
\begin{example}
    Let the Markov chain $(S,T)-X-Y$ hold and the leakage matrix $P_{S|X}\in \mathbb{R}^{|\mathcal{S}|\times|\mathcal{X}|}$ be invertible where the columns of $P_{S|X}$ are conditional distributions $P_{S|X=x}\in\mathbb{R}^{|\mathcal{S}|}$. In this case, using the Markov chain $S-X-Y$, we have
    \begin{align*}
        P_{X|Y=y}=P_{S|X}^{-1}P_{S|Y=y}\stackrel{(a)}{=}P_{S|X}^{-1}P_{S}\stackrel{(b)}{=}P_X,
    \end{align*}
    where (a) follows by the independence of $S$ and $Y$, (b) holds since by using the Markov chain $S-X-Y$ we have $P_{S|X}P_X=P_S$. Hence, in this case, $I(S;Y)=0$ leads to $I(Y;X)=0$. Furthermore, using the Markov chain $T-X-Y$, we have
    \begin{align*}
    P_{T|Y=y}&=P_{T|X}P_{X|Y=y}\stackrel{(a)}{=}P_{T|X}P_{X}=P_T,
    \end{align*}
    where (a) follows by $I(X;Y)=0$. Thus, we conclude that $I(Y;S)=0$ leads to $I(T;Y)=0$, i.e., zero utility in \eqref{eq:problem_1}.
 \end{example}

%% file: contents/main_results.tex
In this section, we first improve the bounds derived in \cite{AmirITW2024} using a tighter version of SFRL in \cite{new_functional}, i.e., Lemma \ref{lemma2}. We then provide lower and upper bounds for the trade-off in \eqref{eq:problem_1}, under the bounded demographic parity (i.e., non-zero leakage) and a certain encoding rate. To this end, we utilize the extended version of FRL and SFRL, i.e., Lemmas \ref{lemma3} and \ref{lemma4}.

Before presenting the theorems, we first write the utility $I(Y, T)$ in $g^{r}_\epsilon(P_{S,X,T})$ as follows
\begin{align}
I(Y;T)&=I(X,S,T;Y)-I(X,S;Y|T),\nonumber\\%&=I(X,S;Y)+I(T;Y|X,S)-I(X,S;Y|T),\nonumber\\
&=I(X,S;Y)+H(T|X,S)-H(T|Y,X,S)\nonumber\\ &\quad -I(X,S;Y|T).
\label{key}
\end{align}
Equation \eqref{key} helps us to find upper and lower bounds on \eqref{eq:problem_1}. In the next result, compared to \cite[Theorem 1]{AmirITW2024}, we provide new upper bounds and improve the lower bounds $L_3^{r}$ and $L_2$ in \cite[Theorem 1]{AmirITW2024} using Lemma \ref{lemma2}.
%Needs modification
\begin{theorem}
    \label{th_prev}
	For every compression level $0\leq r\leq H(X|S)$ and RVs $(S, X, T) \sim P_{S,X,T}$, we have that
	\begin{align*}
    %\label{th2}
	\max\{L_1^{r},L_2,L_3^{r}\}&\leq g^{r}_{0}(P_{S,X,T})\\&\leq \operatorname{min}\{H(T|X)+r ,H(T|S),U_0\},
	\end{align*}
	where
	\begin{align*}
	L_1^{r} &= H(T|X,S)+r-H(X,S|T),\\
	L_2 &= H(T|X,S)-\left( \log(I(X,S;T)+5.51)+1.06 \right),\\
	L_3^{r} &= H(T|X,S)+r-\alpha H(X,S|T)-1.06\\&\!\!\!\!\!\!\!\! \ -\!\log((1\!-\!\alpha)I(X,S;T)\!+\!\alpha\min\{H(T),H(X,S)\}\!+\!5.51),\\
    U_0&=H(T)+\\&\!\!\!\!\!\sum_{t\in\mathcal{T}}\int_{0}^{1} \mathbb{P}_S\{P_{T|S}(t|S)\geq m\}\log (\mathbb{P}_S\{P_{T|S}(t|S)\geq m\})dm
\end{align*}
and $\alpha=\frac{r}{H(X|S)}$. Moreover, for $H(X|S)\leq r< H(X)$ we have that
 \begin{align}
 \label{th22}
 \max\{{L'}_1^{r},L_2\}&\leq g^{r}(P_{S,X,T})\leq\\& \operatorname{min}\{H(T|X)+r ,H(T|S),U_0\},
 \end{align}
 where $
 {L'}_1^{r} = H(T|S)-H(S|T)=H(T)-H(S).$
\end{theorem}
\begin{proof}
To obtain $L_3^{r}$ and $L_2$, we use arguments similar to those in \cite[Th. 1]{AmirITW2024}. The only difference is that we use Lemma \ref{lemma2} instead of SFRL. To obtain the new upper bound $H(T|X)+r$, by removing the constraint $I(U;S)=0$, we obtain
\begin{align*}
    g^{r}_{0}(P_{S,X,T})\leq \sup_{\begin{array}{c} 
	\substack{I(X;Y)\leq r }
	\end{array}}I(Y;T)\stackrel{(a)}{\leq} H(T|X)+r,
\end{align*}
where (a) follows by \cite[Lemma 6]{king3}. Finally, to obtain the upper bound $U_0$, by removing the rate constraint and using \cite[Th. 4]{king3} with $Y\leftarrow T$ and $X\leftarrow S$, we have
\begin{align*}
g^{r}_{0}(P_{S,X,T})\leq\!\!\! \sup_{\begin{array}{c} 
	\substack{I(S;Y)=0 }
	\end{array}}\!\!\!I(Y;T)\leq U_0.
\end{align*}
\end{proof}
\begin{corollary}
    As argued in \cite{AmirITW2024}, considering the regime $H(X|S)\leq r\leq H(X)$, when $S$ is a deterministic function of $T$, i.e., $S=f(T)$, $L_1^r$ is tight achieving the upper bound $H(T|S)$. Furthermore, considering the regime $0\leq r\leq H(X|S)$, when $S$ is a deterministic function of $X$, and $X$ a deterministic function of $T$, i.e., $S=f(X)$ and $X=g(T)$, $L_1^r$ is tight achieving the new upper bound $H(T|X)+r$, which is presented in this paper. As argued in \cite{king3}, there are many cases where $U_0$ is a tighter bound than $H(T|S)$, i.e., $U_0\leq H(T|S)$. For more detail, see \cite[Example 3]{king3}.  
\end{corollary}
Next, we present the results using the bounded statistical parity constraint, that is, $\epsilon \geq 0$.

\begin{theorem}
    \label{th_new}
	For every compression level $r \geq 0$, statistical parity bound $\epsilon \geq 0$, and RVs $(S, X, T) \sim P_{S,X,T}$, the trade-off denoted as $g^{r}_{\epsilon}(P_{S,X,T})$, can be upper bounded by 
    \begin{align}\label{Upper_asli}
        g^{r}_{\epsilon}(P_{S,X,T})\leq \operatorname{min}\left\{H(T|S)+\epsilon, H(T|X)+r\right\}.
    \end{align} 
    Moreover, the following lower bounds can be attained depending on the regime of parameters:
    \begin{enumerate}
    \item $0 \leq r \leq H(X|S)+\epsilon$: 
	\begin{align}
    %\label{th2}
	\max\{L_0, L_1^{r, \epsilon}, L_2^{r, \epsilon}\}\leq g^{r}_{\epsilon}(P_{S,X,T}),
	\end{align}
    	where
	\begin{align*}
     L_0 &= H(T|X,S)\!-\!\left( \log(I(X,S;T)\!+\!5.51)\!+\!1.06 \right), \\
	L_1^{r, \epsilon} &= H(T|X,S)+r-H(X,S|T),\\
	L_2^{r, \epsilon} &= H(T|X,S)+r-\alpha H(X,S|T)\\&\quad-1.06-\log((1\!-\!\alpha)I(X,S;T)\!\!\\&\quad+\alpha\min\{H(T),H(X,S)\}\!+\!5.51),
\end{align*}
and $\alpha=\frac{r}{H(X|S)+\epsilon}$.
\item $H(X|S)+\epsilon \leq r < H(X)$
\begin{align}
    %\label{th2}
	\max\{L_0, L_1'^{r, \epsilon}\}\leq g^{r}_{\epsilon}(P_{S,X,T}),
	\end{align}
    where
%\begin{align*}
 $   L_1'^{r, \epsilon} = H(T|S)-H(S|T)+\epsilon$.
%\end{align*}
    \item $ r \leq \epsilon \leq H(S|X)+r$:
	\begin{align}
    %\label{th2}
	\max\{L_0, L_3^{r, \epsilon}, L_4^{r, \epsilon}\}\leq g^{r}_{\epsilon}(P_{S,X,T}),
	\end{align}
    	where
\begin{align*}
      L_3^{r, \epsilon} &=  H(T|X,S)+\epsilon-H(X,S|T),\\
    L_4^{r, \epsilon} &=  H(T|X,S)+\epsilon-\tilde{\alpha} H(X,S|T)\\&\quad-1.06-\log((1\!-\!\tilde{\alpha})I(X,S;T)\!\\ &\quad+\tilde{\alpha}\min\{H(T),H(X,S)\}\!+\!5.51),
\end{align*}
and $\tilde{\alpha}=\frac{\epsilon}{H(S|X)+r}$.

\item $H(S|X)+r \leq \epsilon < H(S)$
\begin{align}
    %\label{th2}
	\max\{L_0, L_3'^{r, \epsilon}\}\leq g^{r}_{\epsilon}(P_{S,X,T}),
	\end{align}
    where
%\begin{align*}
 $   L_3'^{r, \epsilon} = H(T|X)-H(X|T)+r$.
%\end{align*}
    \end{enumerate}% for $H(X|S)\leq r< H(X)$ we have that
%  \begin{align}
%  \label{th22}
%  \max\{{L'}_1^{r},L_2\}\leq g^{r}(P_{S,X,T})\leq H(T|S),
%  \end{align}
%  where $
%  {L'}_1^{r} = H(T|S)-H(S|T)=H(T)-H(S).$
\end{theorem}

% \begin{theorem}
%     \label{th.2}
% 	For every leakage level $\epsilon \geq 0$, every compression level $0\leq r\leq H(X|S) + \epsilon$, and random variables $(S, X, T) \sim P_{S,X,T}$, we have that
% 	\begin{align*}
%     %\label{th2}
%     g^{r}_{\epsilon}(P_{S,X,T})\leq \operatorname{min}\left\{H(T|S)+\epsilon, H(T|X)+r\right\} \\
% 	\max\{L_1^{r, \epsilon}, L_2, L_3^{r, \epsilon}\}\leq g^{r}_{\epsilon}(P_{S,X,T}),
% 	\end{align*}
% 	where
% 	\begin{align*}
% 	L_1^{r, \epsilon} &= H(T|X,S)+r-H(X,S|T),\\
% 	L_2 &= H(T|X,S)-\left( \log(I(X,S;T)+1)+4 \right),\\
% 	L_3^{r, \epsilon} &= H(T|X,S)+r-\alpha H(X,S|T)-4\\&\!\!\!\!\!\!\!\! \ -\!\log((1\!-\!\alpha)I(X,S;T)\!+\!\alpha\min\{H(T),H(X,S)\}\!+\!1),
% \end{align*}
% and $\alpha=\nicefrac{r}{H(X|S)+\epsilon}$. Moreover, for the regime $r \leq \epsilon \leq H(S|X)+r$, the following lower band holds

% where
% \begin{align*}
%       L_4^{r, \epsilon} &= H(T|X,S)+\epsilon-H(X,S|T)\\
%     L_5^{r, \epsilon} &= H(T|X,S)+\epsilon-\tilde{\alpha} H(X,S|T)-4\\&\!\!\!\!\!\!\!\! \ -\!\log((1\!-\!\tilde{\alpha})I(X,S;T)\!+\!\tilde{\alpha}\min\{H(T),H(X,S)\}\!+\!1),
% \end{align*}
% and $\Bar{\alpha}=\nicefrac{\epsilon}{H(S|X)+r}$.
% % for $H(X|S)\leq r< H(X)$ we have that
% %  \begin{align}
% %  \label{th22}
% %  \max\{{L'}_1^{r},L_2\}\leq g^{r}(P_{S,X,T})\leq H(T|S),
% %  \end{align}
% %  where $
% %  {L'}_1^{r} = H(T|S)-H(S|T)=H(T)-H(S).$
% \end{theorem}
\begin{sproof}
Here, the main idea of the proof is provided and the complete proof is provided in Appendix A.
%Upperbound
   To derive the upper bounds, we first remove either the parity or rate constraints, and then we use \cite[Lemma 6]{king3}.
   The derivation of $L_0$ is similar to the proof of \cite[Theorem 1]{AmirITW2024} since it does not depend on $r$ and $\epsilon$. In other words, the lower bounds in \cite[Theorem 1]{AmirITW2024} can be also used for $g^{r}_{\epsilon}(P_{S,X,T})$, since we have
   $
   g^{r}_{0}(P_{S,X,T})\leq g^{r}_{\epsilon}(P_{S,X,T}).
   $
   The only difference is that in this paper we use Lemma \ref{lemma2} instead of the SFRL. To obtain $L_1^{r, \epsilon}$, we first construct $U$ using the FRL (Lemma 1) so that $I(U;S)=\epsilon$ and $H(X|U,S)=0$. Since such $U$ does not necessarily satisfy the rate constraint, we randomize over $U$ to ensure that $I(U';X)\leq r$ and $I(U';S)\leq \epsilon$. Finally, we produce $Y'$ using Lemma 1 such that $I(Y';S,X,U')=0$ and $H(T|Y',U',X,S)=0$. Then we let $Y=(Y',U')$ and we show that it achieves $L_1^{r, \epsilon}$. To achieve $L_2^{r,\epsilon}$, we use same $U'$ as before, but to produce $Y'$ we use Lemma 2 conditionally with $X\leftarrow (S,X)|U'$ and $Y\leftarrow T$. To attain $L_3^{r,\epsilon}$, we first build $U$ using Lemma 3 ensuring $I(U;X)=r$ and then randomize it so that $U'$ satisfies $I(U';S)\leq \epsilon$. Finally, we build $Y'$ similar as in achieving $L_1^{r, \epsilon}$. The achievability of $L_4^{r,\epsilon}$ is similar and the only difference is to build $Y'$ by using Lemma 2 conditionally. To achieve $L_1'^{r,\epsilon}$ and $L_3'^{r,\epsilon}$ we use
   %\begin{align*}
    $   L_1'^{r,\epsilon} = L_1^{r=H(X|S)+\epsilon,\epsilon}$ and
       $L_3'^{r,\epsilon} = L_1^{r,\epsilon=H(S|X)+r}$.
   %\end{align*}
   The main reason for producing $Y'$ in the lower bounds is to make the term  $H(T|X,S,Y)$ in \eqref{key} zero, thereby improving the utility. Furthermore, the main reason to produce $U$ in $L_1^{r,\epsilon}$ is that the term $I(X,S;Y)$ in \eqref{key} achieves its maximum which equals to $\epsilon+H(X|S)$, however, it does not necessarily satisfy the rate constraint. Hence, we randomized over $U$ to produce $U'$. A similar reason holds for achieving $L_3^{r,\epsilon}$. 
   \end{sproof}
\begin{figure}
    \centering
    \includegraphics[width=1\linewidth]{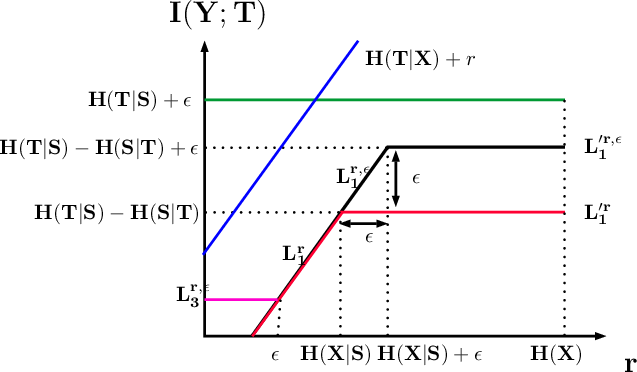}
    \caption{Comparison of the bounds in Theorems 1 and 2.}
    \vspace{-2mm}
    \label{compare}
\end{figure}
% To this end,
% \subsection*{Comparison and Discussion}
\vspace{-5mm}
\begin{remark}\normalfont 
As shown in Fig. \ref{compare}, comparing Theorems 1 and 2, we can see that relaxing the perfect demographic parity constraint leads to improved utility. For instance, 
$L_{1}^{r,\epsilon}$
  improves upon 
  $L_1^r$
  as it is attained over a larger regime.  
  Furthermore, 
$L_{2}^{r,\epsilon}$
  can improve 
$L_3^r$, as the corresponding 
$\alpha$ is smaller compared to Theorem 1. The latter follows since by ignoring the $\log(\cdot)$ terms (since
they have smaller values compared to the other terms), $\alpha H(X,S|T)$ is smaller in $L_{2}^{r,\epsilon}$.
\end{remark}

%\begin{remark}\normalfont

 %   Talking about $\alpha$ and $\tilde{\alpha}$ and how they affect the bound compared to the perfect case
%\end{remark}

\begin{remark}\label{re6}
As we outlined earlier, to achieve $L_1^{r,\epsilon}$ and $L_2^{r,\epsilon}$, we first build $U$ using the EFRL and ESFRL, i.e., Lemmas 3 and 4,  ensuring that $I(U;S)=\epsilon$. We then randomize over $U$ such that $U'$ satisfies $I(U';X)\leq r$.   
    On the other hand, to attain $L_3^{r,\epsilon}$ and $L_4^{r,\epsilon}$, we first build $U$ using Lemmas 3 and 4 ensuring $I(U;X)=r$ and then randomize it so that $U'$ satisfies $I(U';S)\leq \epsilon$. The novelty in achieving $L_3^{r,\epsilon}$ and $L_4^{r,\epsilon}$ follows by using the relaxed statistical parity constraint $I(Y;S)\leq \epsilon$. In other words, since we are allowed to leak information about $S$, we can construct $U$ to attain $\epsilon$ for the parity constraint and then update it to satisfy the rate constraint. Finally, if we enforce perfect demographic parity, we cannot use such a technique.
\end{remark}

%\begin{remark}
 %   Also explain about the upper bounds
%\end{remark}

%In the following, we provide a numerical example across different cases, each highlighting a scenario where one of the lower bounds is stronger.
In the following, we illustrate some of the bounds of Theorem \ref{th_new} using an example inspired by the Experiment section provided in \cite{zamani2025variable}. 
\begin{example}\normalfont
\label{example_numeric}
   \begin{figure}[ht]
    \centering
    \includegraphics[width=1\linewidth]{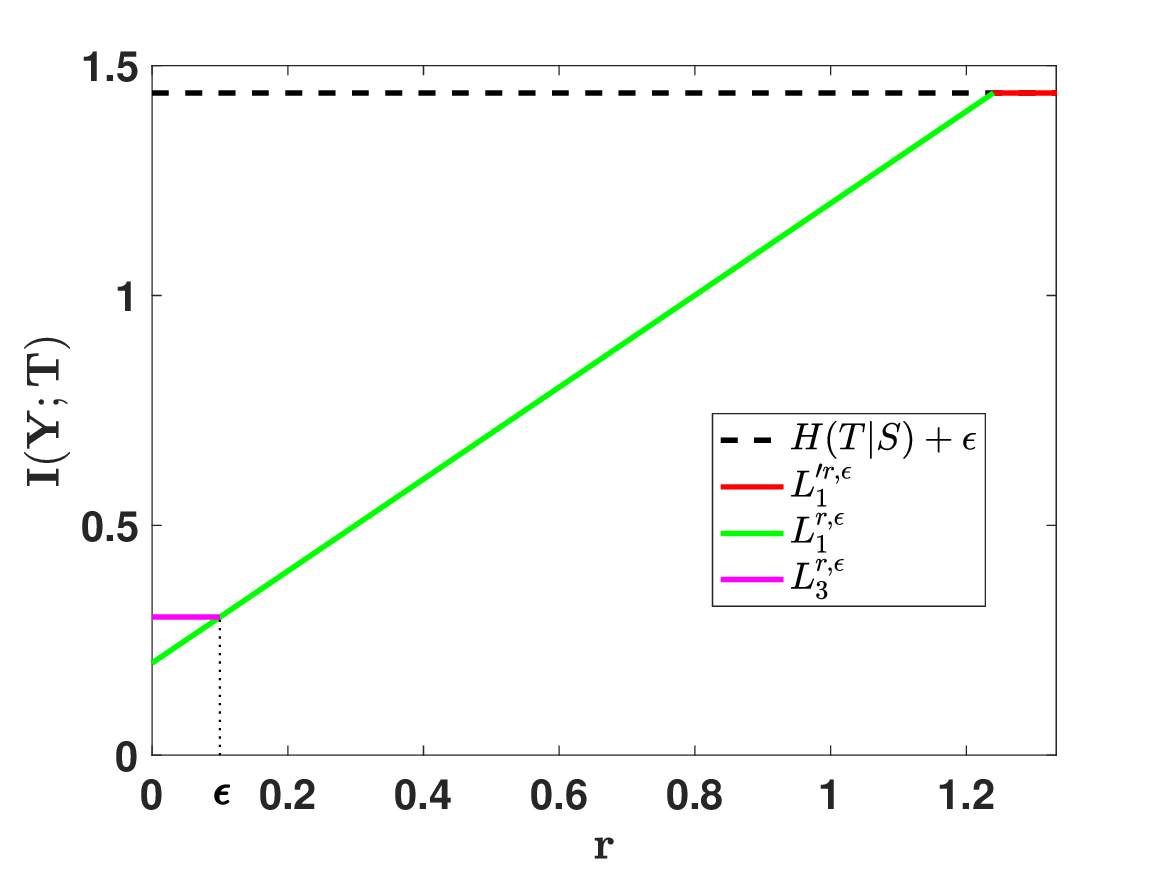}
    \caption{Comparison of the bounds in Example \ref{example_numeric}.}
    \vspace{-2mm}
    \label{compare_ex1}
\end{figure}
Considering \cite[Experiment]{zamani2025variable}, let 
	$X$ denote the useful data defined as a ``quantized histogram." Similar to \cite{zamani2025variable}, we find the quantized histogram by converting the images into black and white and then quantizing the histogram of any image using the following intervals: Interval~1 $=[0,0.1)$, Interval~2 $=[0.1,0.15)$, Interval~3 $=[0.15,0.2)$, Interval~4 $=[0.2,0.25)$, Interval~5 $=[0.25,0.3)$, Interval~6 $=[0.3,0.35)$ and Interval~7 $=[0.35,1]$. Therefore, each image in the database is labeled with a number between 1 and 7, indicating the interval to which its histogram belongs.
    Let the target $T$ be the label of each image that is a number from $0$ to $9$ indicating the digit which the image represents, i.e., $T\in\{0,1,...,9\}$. Furthermore, let the private attribute $S$ be a ternary RV, i.e., $S\in\{0,1,2\}$ and be defiend as $S=\begin{cases}
	2\ \text{if}\ T\in\{4,5,7,9\},\\
	1\ \text{if}\ T\in \{0,2,3,6,8\},\\
    0\ \text{if}\ T\in\{1\}
	\end{cases}$. To find the joint distribution of $(S,X,T)$, we use \cite[Eq. (94)]{zamani2025variable} and \cite[Table. 1]{zamani2025variable}, substituting $Z$ with $T$, $Y$ with $X$ and $X$ with $S$. Here, we choose $\epsilon=0.1$. In Fig. \ref{compare}, the lower bounds $L_1^{r,\epsilon}$, $L_1'^{r,\epsilon}$, $L_3^{r,\epsilon}$, and the upper bound $H(T|S)+\epsilon$ are shown. Clearly, in this example $S$ is a deterministic function of $T$, therefore; the lower bound $L_1'^{r,\epsilon}$, meets the upper bound $H(T|S)+\epsilon$. This can be verified in Fig. \ref{compare}. Furthermore, for $r\leq \epsilon$, we can see that $L_3^{r,\epsilon}$ improves $L_1^{r,\epsilon}$.
\end{example}
Next, we provide an example inspired by the noisy typewriter in \cite{cover1999elements}, where the lower bound $L_2^{r,\epsilon}$ improves $L_1^{r,\epsilon}$.
\begin{example}\label{chooneabol}\normalfont
Considering the noisy typewriter example in \cite{cover1999elements}, let $T$ correspond to the input alphabet and $X$ be the output alphabet, where $X$ is a noisy version of $T$. 
Here, let $T$ be a uniform RV over $\{1, \ldots, 1000\}$. A typewriter produces the output by choosing the same input symbol with probability $\frac{1}{3}$ and 
any other character uniformly, i.e., with probability $\frac{2}{3 \times 999}$ for each of the remaining characters. In other words, we have
\begin{align*}
    P_{X|T=t}=\begin{cases}
    X=t\ \text{w.p.}\ \frac{1}{3},\\
    X=i\ \text{w.p.}\ \frac{2}{3*999}, \ \forall i\in\{1,...,1000\}\backslash t,
    \end{cases}.
\end{align*}
Furthermore, let the private attribute $S$ be a ternary RV, i.e., $S \in \{0,1,2\}$, and define it as  
\[
S = \begin{cases}  
    2, & \text{if } \text{res}_{10} (T) \in \{2,\ldots,9\},\\  
    1, & \text{if } \text{res}_{10} (T) \in \{1\},\\  
    0, & \text{if } \text{res}_{10} (T) \in \{0\},  
\end{cases}  
\]
where $\text{res}_{10} (T)$ denotes the residue of $T$ modulo 10. Clearly, in this example, $S$ is a deterministic function of $T$. Here, we choose $\epsilon=0.1$. As shown in Fig.~\ref{compare5}, for the regime $r \leq 2.52$, $L_2^{r,\epsilon}$ improves upon $L_1^{r,\epsilon}$. 
The lower bound $L_3^{r,\epsilon}$ becomes negative and is therefore omitted. 
Furthermore, for $r \leq 1.08$, the upper bound $H(T|X) + r$ is 
tighter than the upper bound $H(T|S) + \epsilon$. Finally, similar to the previous example, the lower bound $L_1'^{r,\epsilon}$ is tight since $S = f(T)$.
\begin{figure}[ht]
    \centering
    \includegraphics[width=1\linewidth]{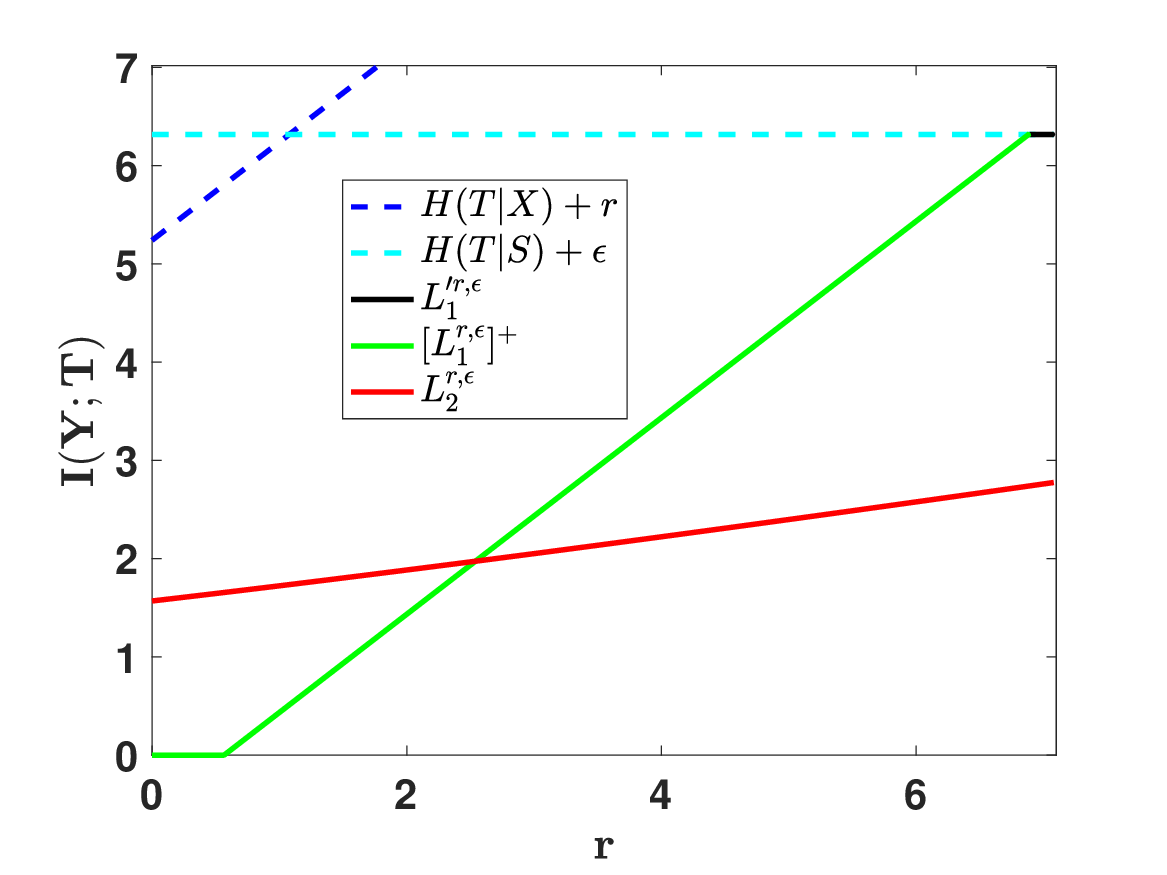}
    \caption{Comparison of the bounds in Example \ref{chooneabol}. Here, $[L_{1}^{r, \epsilon}]^+$ corresponds to $\max\{0, L_{1}^{r, \epsilon}\}$.}
    \vspace{-3mm}
    \label{compare5}
\end{figure}

\end{example}

%% file: contents/conclusion.tex
In this paper, we studied the design of fair/private representations of a database under a compression rate constraint. We argued that, since perfect privacy or perfect demographic parity is not always feasible, we extended the fairness/privacy constraint to allow a bounded statistical parity/privacy leakage constraint. 
More precisely, we have considered a trade-off between utility $I(Y;T)$, the bounded statistical parity/privacy leakage constraint $I(Y;S)\leq \epsilon$, and the bounded compression rate condition $I(Y;X)\leq r$ where the goal is to maximize the utility under the two constraints.

To find the representations, we used randomization techniques on top of the Functional Representation Lemma and the Strong Functional Representation Lemma, which led to the lower bounds on the trade-off. This randomization approach has also been used to extend the FRL and SFRL. A benefit of using the FRL, SFRL, and their extended versions is that they are constructive and have low complexity. Therefore, the combination of this randomization technique with the lemmas leads to simple representation designs. We have shown that when the compression rate $r$ is greater than a threshold the proposed designs can be optimal. Furthermore, we have shown that using the bounded statistical parity constraint, when the compression rate is less than $\epsilon$, the new bounds can improve the previous bounds. Finally, we studied the bounds in a numerical example and showed that when the private data $S$ is a deterministic function of the target  $T$, the lower bounds meet the upper bound.

%% file: contents/Appendix.tex
\textit{Proof of Theorem \ref{th_new}:}
\label{Appendix_A}
The upper bound is achieved by removing one constraint from the trade-off at a time while keeping the other. To this end, we first omit the rate constraint. Therefore, we have
\begin{align}\label{first_upper}
    g^{r}_{\epsilon}(P_{S,X,T})
    \!\leq \!\!\!\!\!\sup_{\begin{array}{c} 
	\substack{P_{Y|S,X,T}:I(Y;S) \!\leq\! \epsilon}
	\end{array}}\!\!\!\!\!\!I(Y;T) \stackrel{(a)}{\leq} H(T|S) + \epsilon,
\end{align} 
where $(a)$ holds as
\begin{align*}
    I(Y ; T) &= H(T|S) + I(S; Y) - H(T|Y, S) - I(S; Y|T) \\&\leq H(T|S)+\epsilon.
\end{align*}
Next, by removing the bounded
statistical parity and/or privacy leakage constraint, we can write
\begin{align}\label{second_upper}
    g^{r}_{\epsilon}(P_{S,X,T})
    \!\leq\!\!\!\!\!\! \sup_{\begin{array}{c} 
	\substack{P_{Y|S,X,T}:I(X;Y) \leq r}
	\end{array}}\!\!\!\!\!\!\!\!I(Y;T) \!\stackrel{(a)}{\leq}\! H(T|X) + r,
\end{align} 
where $(a)$ holds as
\begin{align*}
    I(Y ; T) &= H(T|X) + I(X; Y) - H(T|Y, X) - I(X; Y|T) \\&\leq H(T|X) + r.
\end{align*}
 Finally, the upper bound \eqref{Upper_asli} follows from the combination of \eqref{first_upper} and \eqref{second_upper}.

To obtain $L_0$, we first note that it holds in all regimes. Using Lemma \ref{lemma2} considering $U \leftarrow Y$, $X \leftarrow (X, S)$ and $Y \leftarrow T$, we have 
	\begin{align*}
	I(T;Y)&=H(T|X,S)-I(X,S;Y|T)\\&\geq H(T|X,S)\!-\!\left( \log(I(X,S;T)+5.51)+1.06 \right)\\&=L_0.
	\end{align*} 

For the lower bound $ L_1^{r, \epsilon}$, we use Lemma \ref{lemma3} with
$U \leftarrow U$, $X \leftarrow S$, and $Y \leftarrow X$ such that 
\begin{align}
I(S;U)& = \epsilon,\label{ss}\\
H(X|U,S)&=0. \label{ss1}
\end{align}
Moreover, we have 
\begin{align}
    I(U;X) \leq H(X|S)+\epsilon,\label{tt}
\end{align}
where~\eqref{tt} holds since
\begin{align}
I(U;X)&\stackrel{(a)}{=}I(U;S)+H(X|S)-H(X|U,S)-I(S;U|X)\nonumber\\&\stackrel{(b)}{\leq} \epsilon+H(X|S).
\end{align}
where (a) is valid for any correlated $X$, $S$, and $U$ and (b) is due to \eqref{ss} and \eqref{ss1}. We now employ the randomization technique to construct $U'$ from $U$. 

Let
$
U'=\begin{cases}
U,\ \text{w.p}.\ {\alpha}\\
c,\ \ \text{w.p.}\ 1-{\alpha}
\end{cases},
$
\looseness=-1 where $c$ is a constant which does not belong to $\mathcal{X}\ \cup \ \mathcal{S} \ \cup \ \mathcal{U}$ and $\alpha=\frac{r}{H(X|S)+\epsilon}$. Note that we have $0 \leq \alpha \leq 1 $ in the given regime since $0 \leq r \leq H(X|S)+\epsilon$. Furthermore, assume $Y'$ be the output of Lemma \ref{lemma1} with $X\leftarrow (S,X,U')$ and $Y\leftarrow T$. Hence, we can write 
\begin{align}
I(Y';S,X,U')&=0,  \label{sss}\\%\textnormal{ and}
H(T|Y',S,X,U')&=0.\label{ssss}
\end{align}
Now, it is sufficient to let $Y=(U',Y')$. We first verify that the constraints of
$I(S;Y) \leq \epsilon$ and $I(X;Y)\leq r$
are satisfied. To this end,
%
%\begin{align*}
$I(S;Y)=I(S;Y',U')\stackrel{(a)}{=}I(S;U')=\alpha I(S;U)= \alpha \epsilon \leq \epsilon,$ 
%\stackrel{(b)} 
%\end{align*}
where (a) follows by \eqref{sss}. Moreover, we have
%
%\begin{align*}
$I(X;Y)=I(X;Y',U')\stackrel{(a)}{=}I(X;U')=\alpha I(X;U) \stackrel{(b)}{\leq} \alpha (H(X|S)+\epsilon)=r,$
%\end{align*}
where (a) follows from \eqref{sss} and (b) follows from \eqref{tt}. Next, it follows from \eqref{key} that
\begin{align*}
I(T;Y)&=I(T;Y',U') 
\\&\stackrel{(a)}{=}H(T|X,S)\!+\!I(Y',U';X,S)\!-\!I(X,S;Y'\!,U'|T) 
\\&\stackrel{(b)}{=} H(T|X,S)+I(U';X,S)-I(X,S;Y',U'|T) 
 \\&\geq H(T|X,S)+I(U';X,S)-H(X,S|T)
%\\&=H(T|X,S)+I(U';S)+I(U';X|S)-H(X,S|T)
\\&\stackrel{(c)}{=}H(T|X,S)+r-H(X,S|T)
\end{align*}
where we applied \eqref{ssss} and \eqref{sss} in (a) and (b), respectively. Furthermore, (c) follows as
\begin{align}\label{U_prime_X_S}
I(U';X,S) &= I(U';S)+I(U';X|S)\nonumber\\&=\alpha(I(U;S) + I(U;X|S))\nonumber\\&\stackrel{(a)}{=}\alpha (\epsilon+H(X|S))=r,
\end{align}
where we used \eqref{ss1} in (a). 
To derive $L_2^{r, \epsilon}$, we use the same $U'$ as in deriving $L_1^{r, \epsilon}$, but to produce $Y'$, we use Lemma \ref{lemma2} with $X\leftarrow (S,X)|U'$ and $Y\leftarrow T$. 
In this case, in addition to \eqref{ssss},
%in addition to XXX for $U'$,
%\eqref{} %$U'$ satisfies \eqref{sss}, \eqref{ssss} and
 the following bound for $Y'$ holds
\begin{align}
\label{first_bound_y_prime}
I(&Y';S,X|T,U') \nonumber\\
&\leq \log(I(X,S;T|U')+5.51)+1.06
 \nonumber\\&= \log((1-\alpha)I(X,S;T)+\alpha I(X,S;T|U)\!+\!5.51)\!+\!1.06
 \nonumber\\&\leq \log((1\!-\!\alpha)I(X,S;\!T)\!+\!\alpha\min\{H(X,S),H(T)\}\!+\!5.51\!)\! \nonumber\\&\quad+\!1.06.
\end{align}	
%\vspace{-2mm}
Furthermore, $I(X,S;Y',U'|T)$ can be bounded as
\vspace{-1mm}
\begin{align}
I(&X,S;Y',U'|T)\nonumber \\
&=I(X,S;U'|T)+I(X,S;Y'|T,U')\nonumber\\
&=\alpha I(X,S;U|T)+I(X,S;Y'|T,U') \nonumber \\
&\leq \alpha H(X,S|T)+\log((1\!-\!\alpha)I(X,S;\!T) \nonumber\\ & \quad +\!\alpha\min\{H(X,S),H(T)\}\!+\!5.51\!)\!+\!1.06,\label{t}
\end{align}
where \eqref{t} follows from \eqref{first_bound_y_prime}.
Now, using \eqref{key}, \eqref{ssss}, \eqref{U_prime_X_S}, and \eqref{t}, we have $I(Y;T) \geq L_2^{r, \epsilon}$.
%\begin{align*}
% $   L_3^r = H(T|X,S)+r-\alpha H(X,S|T) 
%    -\log((1\!-\!\alpha)I(X,S;\!T)\! 
%    +\!\alpha\min\{H(X,S),H(T)\}\!+\!1\!)\!+\!4.$
%\end{align*}
% \begin{align*}
% &I(T;Y)\geq H(T|X,S)+r-\alpha H(X,S|T)-\\ &\log((1\!-\!\alpha)I(X,S;\!T)\!+\!\alpha\min\{H(X,S),H(T)\}\!+\!1\!)\!+\!4\\&=L_3^{r}.
% \end{align*}
To obtain ${L'}_1^{r, \epsilon}$ in the region $H(X|S)+\epsilon\leq r< H(X)$, note that
\begin{align*}
g^{r}_{\epsilon}(P_{S,X,T})&\geq g^{H(X|S)+\epsilon}_{\epsilon}(P_{S,X,T})\\&\geq L_1^{H(X|S)+\epsilon, \epsilon}\\&=H(T|X,S)+H(X|S)+\epsilon-H(X,S|T)\\&= H(T,X|S)-H(S,X|T)+\epsilon\\&= H(T|S)-H(S|T)+\epsilon.
\end{align*}

To derive the lower bound ${L}_3^{r, \epsilon}$, we use Lemma \ref{lemma3} with $U \leftarrow U$, $X \leftarrow X$ and $Y \leftarrow S$ such that
\begin{align}
I(U;X)& = r,\label{nn1}\\
H(S|U,X)&=0. \label{nn2}
\end{align}
Moreover, we have 
\begin{align}
    I(U;S)\leq H(S|X)+r,\label{nn3}
\end{align}
where \eqref{nn3} is valid since 
\begin{align}\label{dual_bound}
I(U;S)&\stackrel{(a)}{=}I(U;X)+H(S|X)-H(S|U,X)-I(X;U|S)\nonumber\\&\stackrel{(b)}{\leq} r+H(S|X).
\end{align}
where (a) holds for any correlated $X$, $S$, and $U$ and (b) is due to \eqref{nn1} and \eqref{nn2}. We now employ the randomization technique to construct $U'$ from $U$. Let 
$
U'=\begin{cases}
U,\ \text{w.p}.\ \tilde{\alpha}\\
\tilde{c},\ \ \text{w.p.}\ 1-\tilde{\alpha}
\end{cases},
$
\looseness=-1 where $\tilde{c}$ is a constant which does not belong to $\mathcal{X} \cup \mathcal{S} \cup \mathcal{U}$ and $\tilde{\alpha}=\frac{\epsilon}{H(X|S)+r}$. It is worth mentioning that we have $0 \leq \tilde{\alpha} \leq 1 $ in the given regime since $r \leq  \epsilon \leq H(S|X)+r$. Furthermore, assume $Y'$ is constructed as in deriving $L_{1}^{r, \epsilon}$. Therefore, \eqref{sss} and \eqref{ssss} hold. 
%the output of Lemma \ref{lemma1} with $X\leftarrow (S,X,U')$ and $Y\leftarrow T$. Hence, we can write 
%\begin{align}
%I(Y';S,X,U')&=0,  \label{nnn}\\%\textnormal{ and}
%H(T|Y',S,X,U')&=0.\label{nnnn}
%\end{align}
‌Now, it is sufficient to let $Y=(U',Y')$. First, we show the validity of the constraints
$I(S;Y) \leq \epsilon$ and $I(X;Y)\leq r$. To this end,
%
%\begin{align*}
$I(S;Y)=I(S;Y',U')\stackrel{(a)}{=}I(S;U')= \tilde{\alpha} I(S;U) \stackrel{(b)}{\leq} \tilde{\alpha} (H(S|X)+r) \leq \epsilon,$ 
%\stackrel{(b)} 
%\end{align*}
where (a) and (b) follow from \eqref{sss} and \eqref{dual_bound}, respectively. Moreover, we have
%
%\begin{align*}
$I(X;Y)=I(X;Y',U')\stackrel{(a)}{=}I(X;U')=\tilde{\alpha} I(X;U)\stackrel{(b)}{=} \tilde{\alpha} r \leq r,$
%\end{align*}
where (a) follows from \eqref{sss} and (b) follows from \eqref{nn1}. Next, it follows from \eqref{key} that
\begin{align*}
I(T;Y)&=I(T;Y',U') 
\\&\stackrel{(a)}{=}H(T|X,S)\!+\!I(Y',U';X,S)\!-\!I(X,S;Y'\!,U'|T) 
\\&\stackrel{(b)}{=} H(T|X,S)+I(U';X,S)-I(X,S;Y',U'|T) 
 \\&\geq H(T|X,S)+I(U';X,S)-H(X,S|T)
%\\&=H(T|X,S)+I(U';S)+I(U';X|S)-H(X,S|T)
\\&\stackrel{(c)}{=}H(T|X,S)+\epsilon-H(X,S|T)
\end{align*}
where we applied \eqref{ssss} and \eqref{sss} in (a) and (b), respectively. Moreover, (c) follows as
\begin{align}\label{U_prime_X_S_new}
&I(U';X,S) = I(U';X)+I(U';S|X)=\tilde{\alpha}I(U;X) \nonumber\\&+\tilde{\alpha}I(U;S|X)=\tilde{\alpha} (r+H(S|X))=\epsilon,
\end{align}
where we used \eqref{nn2}.
To derive $L_4^{r, \epsilon}$, we use the same $U'$ as in deriving $L_3^{r, \epsilon}$, but to produce $Y'$, we use Lemma \ref{lemma2} with $X\leftarrow (S,X)|U'$ and $Y\leftarrow T$. 
In this case, in addition to \eqref{ssss},
%in addition to XXX for $U'$,
%\eqref{} %$U'$ satisfies \eqref{sss}, \eqref{ssss} and
 the following bound for $Y'$ holds
\begin{align}
\label{new_bound_y_prime_}
I(&Y';S,X|T,U')
\leq \log(I(X,S;T|U')+5.51)+1.06
 \nonumber\\&= \log((1-\tilde{\alpha})I(X,S;T)+\tilde{\alpha} I(X,S;T|U)\!+\!5.51)\!+\!1.06
 \nonumber\\&\leq \log((1\!-\!\tilde{\alpha})I(X,S;\!T)\!+\!\tilde{\alpha}\min\{H(X,S),H(T)\}\!+\!5.51\!)\! \nonumber\\&\quad+\!1.06.
\end{align}	
Furthermore, $I(X,S;Y',U'|T)$ can be bounded as
\begin{align}
I(&X,S;Y',U'|T)=I(X,S;U'|T)+I(X,S;Y'|T,U')\nonumber\\
&=\tilde{\alpha} I(X,S;U|T)+I(X,S;Y'|T,U') \nonumber \\
&\leq \tilde{\alpha} H(X,S|T)+\log((1\!-\!\tilde{\alpha})I(X,S;\!T) \nonumber\\ & \quad +\!\tilde{\alpha}\min\{H(X,S),H(T)\}\!+\!5.51\!)\!+\!1.06,\label{t_new}
\end{align}
where \eqref{t_new} follows from \eqref{new_bound_y_prime_}.
Now, using \eqref{key}, \eqref{ssss}, \eqref{U_prime_X_S}, and \eqref{t_new}, we have $g^{r}_{\epsilon} \geq L_4^{r, \epsilon}$.
Finally, to obtain ${L'}_3^{r, \epsilon}$ in the region $H(S|X)+r \leq \epsilon < H(S)$, note that
\begin{align*}
&g^{r}_{\epsilon}(P_{S,X,T})\geq g^{r}_{H(S|X)+r}(P_{S,X,T})\geq L_3^{r, H(S|X)+r}\\&=H(T|X,S)+H(S|X)+\epsilon-H(X,S|T)= H(T,S|X)\\&-H(X,S|T)+r= H(T|X)-H(X|T)+r.
\end{align*}

%% file: Wiopt2025.bbl
\begin{thebibliography}{10}
\providecommand{\url}[1]{#1}
\csname url@samestyle\endcsname
\providecommand{\newblock}{\relax}
\providecommand{\bibinfo}[2]{#2}
\providecommand{\BIBentrySTDinterwordspacing}{\spaceskip=0pt\relax}
\providecommand{\BIBentryALTinterwordstretchfactor}{4}
\providecommand{\BIBentryALTinterwordspacing}{\spaceskip=\fontdimen2\font plus
\BIBentryALTinterwordstretchfactor\fontdimen3\font minus
  \fontdimen4\font\relax}
\providecommand{\BIBforeignlanguage}[2]{{%
\expandafter\ifx\csname l@#1\endcsname\relax
\typeout{** WARNING: IEEEtran.bst: No hyphenation pattern has been}%
\typeout{** loaded for the language `#1'. Using the pattern for}%
\typeout{** the default language instead.}%
\else
\language=\csname l@#1\endcsname
\fi
#2}}
\providecommand{\BIBdecl}{\relax}
\BIBdecl

\bibitem{AmirITW2024}
A.~Zamani, B.~Rodr{\'\i}guez-G{\'a}lvez, and M.~Skoglund, ``On information
  theoretic fairness: Compressed representations with perfect demographic
  parity,'' in \emph{2024 IEEE Information Theory Workshop (ITW)}, 2024, pp.
  25--30.

\bibitem{vari}
B.~Rodr{\'\i}guez-G{\'a}lvez, R.~Thobaben, and M.~Skoglund, ``A variational
  approach to privacy and fairness,'' in \emph{2021 IEEE Information Theory
  Workshop (ITW)}.\hskip 1em plus 0.5em minus 0.4em\relax IEEE, 2021, pp. 1--6.

\bibitem{gun}
S.~{Sreekumar} and D.~{G\"{u}nd\"{u}z}, ``Optimal privacy-utility trade-off
  under a rate constraint,'' in \emph{2019 IEEE International Symposium on
  Information Theory}, July 2019, pp. 2159--2163.

\bibitem{zhao2022}
H.~Zhao and G.~J. Gordon, ``Inherent tradeoffs in learning fair
  representations,'' \emph{Journal of Machine Learning Research}, vol.~23,
  no.~57, pp. 1--26, 2022.

\bibitem{zhao2019}
H.~Zhao, A.~Coston, T.~Adel, and G.~J. Gordon, ``Conditional learning of fair
  representations,'' in \emph{International Conference on Learning
  Representations}, 2019.

\bibitem{zemel}
R.~Zemel, Y.~Wu, K.~Swersky, T.~Pitassi, and C.~Dwork, ``Learning fair
  representations,'' in \emph{International conference on machine
  learning}.\hskip 1em plus 0.5em minus 0.4em\relax PMLR, 2013, pp. 325--333.

\bibitem{gronowski2023classification}
A.~Gronowski, W.~Paul, F.~Alajaji, B.~Gharesifard, and P.~Burlina,
  ``Classification utility, fairness, and compactness via tunable information
  bottleneck and r{\'e}nyi measures,'' \emph{IEEE Transactions on Information
  Forensics and Security}, 2023.

\bibitem{hardt}
M.~Hardt, E.~Price, and N.~Srebro, ``Equality of opportunity in supervised
  learning,'' \emph{Advances in neural information processing systems},
  vol.~29, 2016.

\bibitem{king3}
A.~Zamani, T.~J. Oechtering, and M.~Skoglund, ``On the privacy-utility
  trade-off with and without direct access to the private data,'' \emph{IEEE
  Transactions on Information Theory}, vol.~70, no.~3, pp. 2177--2200, 2024.

\bibitem{borz}
B.~{Rassouli} and D.~{G\"{u}nd\"{u}z}, ``On perfect privacy,'' \emph{IEEE
  Journal on Selected Areas in Information Theory}, vol.~2, no.~1, pp.
  177--191, 2021.

\bibitem{khodam}
A.~Zamani, T.~J. Oechtering, and M.~Skoglund, ``A design framework for strongly
  $\chi^2$-private data disclosure,'' \emph{IEEE Transactions on Information
  Forensics and Security}, vol.~16, pp. 2312--2325, 2021.

\bibitem{Khodam22}
{A. Zamani, T. J. Oechtering, and M. Skoglund}, ``Data disclosure with non-zero
  leakage and non-invertible leakage matrix,'' \emph{IEEE Transactions on
  Information Forensics and Security}, vol.~17, pp. 165--179, 2022.

\bibitem{Yanina1}
Y.~Y. Shkel, R.~S. Blum, and H.~V. Poor, ``Secrecy by design with applications
  to privacy and compression,'' \emph{IEEE Transactions on Information Theory},
  vol.~67, no.~2, pp. 824--843, 2021.

\bibitem{makhdoumi}
A.~Makhdoumi, S.~Salamatian, N.~Fawaz, and M.~M{\'e}dard, ``From the
  information bottleneck to the privacy funnel,'' in \emph{2014 IEEE
  Information Theory Workshop}, 2014, pp. 501--505.

\bibitem{yamamoto1988rate}
H.~Yamamoto, ``A rate-distortion problem for a communication system with a
  secondary decoder to be hindered,'' \emph{IEEE Transactions on Information
  Theory}, vol.~34, no.~4, pp. 835--842, 1988.

\bibitem{sankar}
L.~Sankar, S.~R. Rajagopalan, and H.~V. Poor, ``Utility-privacy tradeoffs in
  databases: An information-theoretic approach,'' \emph{IEEE Transactions on
  Information Forensics and Security}, vol.~8, no.~6, pp. 838--852, 2013.

\bibitem{Calmon2}
H.~{Wang}, L.~{Vo}, F.~P. {Calmon}, M.~{M\'{e}dard}, K.~R. {Duffy}, and
  M.~{Varia}, ``Privacy with estimation guarantees,'' \emph{IEEE Transactions
  on Information Theory}, vol.~65, no.~12, pp. 8025--8042, Dec 2019.

\bibitem{fairAsoodeh}
H.~Ghoukasian and S.~Asoodeh, ``Differentially private fair binary
  classifications,'' in \emph{2024 IEEE International Symposium on Information
  Theory (ISIT)}, 2024, pp. 611--616.

\bibitem{shannon}
C.~E. Shannon, ``Communication theory of secrecy systems,'' \emph{The Bell
  system technical journal}, vol.~28, no.~4, pp. 656--715, 1949.

\bibitem{edwards2016censoring}
H.~Edwards and A.~Storkey, ``Censoring representations with an adversary,'' in
  \emph{International Conference on Learning Representations (ICLR)}, 2016, pp.
  1--14.

\bibitem{creager2019flexibly}
E.~Creager, D.~Madras, J.-H. Jacobsen, M.~Weis, K.~Swersky, T.~Pitassi, and
  R.~Zemel, ``Flexibly fair representation learning by disentanglement,'' in
  \emph{International conference on machine learning}.\hskip 1em plus 0.5em
  minus 0.4em\relax PMLR, 2019, pp. 1436--1445.

\bibitem{louizos2015variational}
C.~Louizos, K.~Swersky, Y.~Li, M.~Welling, and R.~Zemel, ``The variational fair
  autoencoder,'' \emph{International Conference on Learning Representations
  (ICLR)}, 2016.

\bibitem{gupta2021controllable}
U.~Gupta, A.~M. Ferber, B.~Dilkina, and G.~Ver~Steeg, ``Controllable guarantees
  for fair outcomes via contrastive information estimation,'' in
  \emph{Proceedings of the AAAI Conference on Artificial Intelligence},
  vol.~35, no.~9, 2021, pp. 7610--7619.

\bibitem{de2022funck}
J.~M. de~Freitas and B.~C. Geiger, ``Funck: Information funnels and bottlenecks
  for invariant representation learning,'' \emph{arXiv preprint
  arXiv:2211.01446}, 2022.

\bibitem{yamamoto}
H.~Yamamoto, ``A source coding problem for sources with additional outputs to
  keep secret from the receiver or wiretappers (corresp.),'' \emph{IEEE
  Transactions on Information Theory}, vol.~29, no.~6, pp. 918--923, 1983.

\bibitem{zamani2025variable}
A.~Zamani and M.~Skoglund, ``Variable-length coding with zero and non-zero
  privacy leakage,'' \emph{Entropy}, vol.~27, no.~2, p. 124, 2025.

\bibitem{9457633}
T.-Y. Liu and I.-H. Wang, ``Privacy-utility tradeoff with nonspecific tasks:
  Robust privatization and minimum leakage,'' in \emph{2020 IEEE Information
  Theory Workshop (ITW)}, 2021, pp. 1--5.

\bibitem{zamani2023multi}
A.~Zamani, T.~J. Oechtering, and M.~Skoglund, ``Multi-user privacy mechanism
  design with non-zero leakage,'' in \emph{2023 IEEE Information Theory
  Workshop (ITW)}.\hskip 1em plus 0.5em minus 0.4em\relax IEEE, 2023, pp.
  401--405.

\bibitem{zarab2}
M.~A. Zarrabian, N.~Ding, and P.~Sadeghi, ``On the lift, related privacy
  measures, and applications to privacy--utility trade-offs,'' \emph{Entropy},
  vol.~25, no.~4, p. 679, 2023.

\bibitem{amircache}
A.~Zamani, T.~J. Oechtering, D.~G{\"u}nd{\"u}z, and M.~Skoglund, ``Cache-aided
  private variable-length coding with zero and non-zero leakage,'' in
  \emph{2023 21st International Symposium on Modeling and Optimization in
  Mobile, Ad Hoc, and Wireless Networks (WiOpt)}, 2023, pp. 247--254.

\bibitem{cheuk}
C.~T. Li and A.~El~Gamal, ``Strong functional representation lemma and
  applications to coding theorems,'' \emph{IEEE Transactions on Information
  Theory}, vol.~64, no.~11, pp. 6967--6978, 2018.

\bibitem{warner1965randomized}
S.~L. Warner, ``Randomized response: A survey technique for eliminating evasive
  answer bias,'' \emph{Journal of the American Statistical Association},
  vol.~60, no. 309, pp. 63--69, 1965.

\bibitem{new_functional}
\BIBentryALTinterwordspacing
C.~T. Li, ``Discrete layered entropy, conditional compression and a tighter
  strong functional representation lemma,'' 2025. [Online]. Available:
  \url{https://arxiv.org/abs/2501.13736}
\BIBentrySTDinterwordspacing

\bibitem{erlingsson2014rappor}
{\'U}.~Erlingsson, V.~Pihur, and A.~Korolova, ``Rappor: Randomized aggregatable
  privacy-preserving ordinal response,'' in \emph{Proceedings of the 2014 ACM
  SIGSAC conference on computer and communications security}, 2014, pp.
  1054--1067.

\bibitem{asoodeh2021three}
S.~Asoodeh, J.~Liao, F.~P. Calmon, O.~Kosut, and L.~Sankar, ``Three variants of
  differential privacy: Lossless conversion and applications,'' \emph{IEEE
  Journal on Selected Areas in Information Theory}, vol.~2, no.~1, pp.
  208--222, 2021.

\bibitem{alghamdi2023saddle}
W.~Alghamdi, J.~F. Gomez, S.~Asoodeh, F.~Calmon, O.~Kosut, and L.~Sankar, ``The
  saddle-point method in differential privacy,'' in \emph{International
  Conference on Machine Learning}.\hskip 1em plus 0.5em minus 0.4em\relax PMLR,
  2023, pp. 508--528.

\bibitem{dwork2010boosting}
C.~Dwork, G.~N. Rothblum, and S.~Vadhan, ``Boosting and differential privacy,''
  in \emph{2010 IEEE 51st annual symposium on foundations of computer
  science}.\hskip 1em plus 0.5em minus 0.4em\relax IEEE, 2010, pp. 51--60.

\bibitem{zamani2023private}
A.~Zamani, T.~J. Oechtering, and M.~Skoglund, ``Private variable-length coding
  with non-zero leakage,'' in \emph{2023 IEEE International Workshop on
  Information Forensics and Security (WIFS)}.\hskip 1em plus 0.5em minus
  0.4em\relax IEEE, 2023, pp. 1--6.

\bibitem{cover1999elements}
T.~M. Cover, \emph{Elements of information theory}.\hskip 1em plus 0.5em minus
  0.4em\relax John Wiley \& Sons, 1999.

\end{thebibliography}
